\documentclass[conference]{IEEEtran}
\IEEEoverridecommandlockouts

\usepackage{amsmath,amssymb,amsthm,bm}
\usepackage{mathtools}
\usepackage{graphicx}
\usepackage{booktabs}
\usepackage{multirow}
\usepackage{hyperref}
\usepackage{enumitem}
\usepackage{cite}
\usepackage{float}
\usepackage{braket}
\usepackage{caption}
\usepackage[framemethod=TikZ]{mdframed}
\usepackage{xcolor}

\hypersetup{
  breaklinks=true,
  colorlinks=true,
  linkcolor=blue,
  urlcolor=blue,
  citecolor=blue
}

\graphicspath{{./}}
\newcommand{\one}{\mathbf{1}}
\theoremstyle{plain}
\newtheorem{theorem}{Theorem}[section]
\newtheorem{proposition}{Proposition}
\newtheorem{lemma}[proposition]{Lemma}
\theoremstyle{remark}
\newtheorem{remark}[proposition]{Remark}
\theoremstyle{definition}

\newmdtheoremenv[
    linewidth=0.5pt,
    linecolor=black,
    backgroundcolor=black!3,
    innertopmargin=6pt,
    innerbottommargin=6pt,
    innerleftmargin=8pt,
    innerrightmargin=8pt,
    roundcorner=2pt,
    skipabove=6pt,
    skipbelow=6pt
]{definition}[proposition]{Definition}

\title{Quantum Injection Pathways for \\ Implicit Graph Neural Networks}

\author{
\IEEEauthorblockN{Pengyuan Xu}
\IEEEauthorblockA{
\textit{Dept. of Computer Science}\\
\textit{University of Victoria}\\
 Victoria, Canada\\
\url{pengyuanxu@uvic.ca}}
\and
\IEEEauthorblockN{Tristan Zaborniak}
\IEEEauthorblockA{
\textit{Dept. of Computer Science}\\
\textit{University of Victoria}\\
 Victoria, Canada\\
\url{tristanz@uvic.ca}}
\and
\IEEEauthorblockN{Luis F. Rivera}
\IEEEauthorblockA{
\textit{Dept. of Computer Science}\\
\textit{University of Victoria}\\
 Victoria, Canada\\
\url{rivera@uvic.ca}}
\and
\IEEEauthorblockN{Hausi A. Muller}
\IEEEauthorblockA{
\textit{Dept. of Computer Science}\\
\textit{University of Victoria}\\
 Victoria, Canada\\
\url{hausi@uvic.ca}}
}

\begin{document}
\maketitle
\thispagestyle{plain}
\pagestyle{plain}

\begin{abstract}
Deep Equilibrium Models (DEQs) replace a stack of explicit layers with a single operator whose fixed point defines the output, giving the expressive power of an arbitrarily deep network at the memory cost of a single layer. Quantum Deep Equilibrium Models (QDEQs) bring this idea to quantum machine learning, offering an alternative to Parameterized Quantum Circuits (PQCs), whose depth is limited by hardware coherence and trainability. Here, we introduce, formulate, and compare three ways of coupling a quantum signal to graph DEQs, differing in where the signal enters the fixed-point operator. \textit{Independent} injection computes the quantum signal once per graph and forward fixed-point solve, and holds it fixed throughout the solve. \textit{State-dependent} injection instead recomputes the signal at every solver step and applies it to the current iterate. \textit{Backbone-dependent} injection likewise recomputes at every iteration but applies the signal to the classical backbone's output evaluated at the current iterate. We establish contraction guarantees for each variant under explicit assumptions on the Lipschitz constants of the classical backbone and the quantum signal. On the TU Dortmund graph-classification benchmarks NCI1, PROTEINS, and MUTAG, independent injection achieves the best test accuracy while using fewer forward-solver iterations than both the classical equilibrium baseline and the two dependent variants. 
\end{abstract}

\smallskip


\section{Introduction and Background}\label{sec:introduction}

\subsection{Graph Neural Networks and Implicit Graph Models}\label{sec:intro-gnns}

Graph neural networks (GNNs) are a leading machine learning framework for graph-structured data in chemistry, biology, and relational reasoning \cite{GNNreview18,GNNreview24}. Common message-passing architectures such as GCN, GraphSAGE, and GIN use explicit finite-depth propagation \cite{GCN,GraphSAGE,GIN}; increasing this depth can improve accuracy but often results in growing optimization difficulty \cite{OonoSuzuki2020}. This motivates alternatives that reduce reliance on increasingly deep explicit propagation.

Deep equilibrium models (DEQs) offer one such alternative: they define outputs as fixed points of a single nonlinear operator, trained by implicit differentiation through that fixed point \cite{DEQ}. \textit{Implicit graph models} adapt this idea to graph learning \cite{IGNN}, where the solution of a \textit{graph-dependent} fixed-point equation (instead of the output of an explicit propagation message-passing GNN) plays the role of the final layer node representation. This yields the representational reach of an arbitrarily deep network at the memory cost of a single layer. The corresponding fixed-point map is the \textit{equilibrium operator}, formalized in Definition \ref{def:equilibrium}.

\begin{definition}[Equilibrium operator and equilibrium state]\label{def:equilibrium}

    Fix a graph $G$ with $N$ nodes, a hidden dimension $d_h$, a normalized propagation matrix $\tilde A \in \mathbb{R}^{N\times N}$ derived from the graph adjacency matrix $A$, and an encoder output $H \in \mathbb{R}^{N\times d_h}$ that summarizes the graph's input features. An \textit{equilibrium operator} $\Phi_{\Theta}$ is a parameter-dependent map:
    
    \begin{equation}\label{eq:op-def}
        \Phi_{\Theta}(\,\cdot\,;\tilde A, H) : \mathbb{R}^{N\times d_h} \to \mathbb{R}^{N\times d_h}
    \end{equation}
    
    \noindent where $\Theta$ denotes trainable parameters. A matrix $Z \in \mathbb{R}^{N\times d_h}$ in its domain is a \textit{candidate node-state matrix}: one row per node, each row a $d_h$-dimensional hidden feature vector. An \textit{equilibrium state} $Z^\star$ is a fixed point satisfying:
    
    \begin{equation}\label{eq:fp-def}
        Z^\star = \Phi_{\Theta}(Z^\star;\tilde A, H)
    \end{equation}

\end{definition}

For \textit{graph-level prediction}, the implicit graph pipeline has three stages. First, a \textit{feed-forward encoder} $H = \mathrm{Enc}(A, X_0)$ maps node features $X_0$ and adjacency matrix $A$ to hidden node features $H$. Second, the \textit{equilibrium operator} $\Phi_\Theta(\,\cdot\,;\tilde{A},H)$ of Definition \ref{def:equilibrium} is solved, yielding the node representation $Z^\star$. Finally, a \textit{permutation-invariant readout} pools $Z^\star$ into a graph representation, followed by a classifier. The forward pass combines one encoder evaluation with a fixed-point solve, while the backward pass obtains gradients through the equilibrium operator by implicit differentiation. In practice, $Z^\star$ is approximated by an iterative solver, as in Definition \ref{def:solver}. We therefore treat solver cost as a quantity of interest.

By Definition \ref{def:equilibrium}, input information affects $Z^\star$ only through the dependence of $\Phi_\Theta$ on graph-derived quantities such as $(\tilde A,H)$. Without such dependence, the fixed point would be determined only by parameters and topology, making the output insensitive to node features. IGNNs typically address this through an input-injection term derived from $H$ \cite{IGNN}. In the hybrid setting studied here, we consider injections deriving from a quantum module acting on $H$ or on the iterate $Z^{(t)}$.

\begin{definition}[Fixed-point solver and solver step]\label{def:solver}
    A \textit{fixed-point solver} for $\Phi_\Theta$ generates a sequence $\{Z^{(t)}\}_{t\ge 0} \subset \mathbb{R}^{N\times d_h}$ from an initialization $Z^{(0)}$ via
    
    \begin{equation}\label{eq:solver-step}
        Z^{(t+1)} = \mathcal{S}\bigl(Z^{(t)};\Phi_\Theta,\tilde A,H\bigr)
    \end{equation}
    
    \noindent where $\mathcal{S}$ is either Picard iteration $\mathcal{S}(Z;\Phi_\Theta,\cdot,\cdot) = \Phi_\Theta(Z;\cdot,\cdot)$ or a quasi-Newton acceleration thereof (Anderson acceleration, Broyden's method, \textit{etc.}) \cite{TorchDEQ}. Each application of $\mathcal S$ is a \textit{solver step}. The iteration is stopped when $\|Z^{(t+1)}-Z^{(t)}\|_F$ falls below a set tolerance; the final iterate is taken as $Z^\star$.
\end{definition}

Whatever form these signals take, \textit{where} they are inserted is a major design choice, because $\Phi_\Theta$ can accommodate such a signal in two structurally distinct ways: 

\begin{enumerate}
    \item[\textnormal{(i)}] A \textit{conditioning} term may be computed once from $(\tilde A, H)$ and added to $\Phi_\Theta$ as an iterate-independent bias. Once computed, it remains constant across all solver steps.
    \item[\textnormal{(ii)}] A \textit{residual} term may be added which depends on the current iterate $Z^{(t)}$. It must therefore be recomputed at every solver step, since $Z^{(t)}$ changes from step to step. 
\end{enumerate}

\noindent The two injection routes differ in how the signal interacts with the solver dynamics of Definition \ref{def:solver} and with the contraction properties of $\Phi_\Theta$ \cite{DEQ,MonotoneIGNN}, affecting both well-posedness (\textit{i.e.}, whether $\Phi_\Theta$ is contractive) and solver cost \cite{TorchDEQ}.

\subsection{Quantum Machine Learning and Depth Constraints}\label{sec:intro-qml}

In parallel to these developments in classical implicit modeling, quantum machine learning has developed a broad family of hybrid quantum-classical models built on the \textit{encode--unitary--measure} pattern, in which classical data are encoded into a quantum state by a data-dependent unitary, processed by a parametrized quantum circuit (PQC), and read out through measurements of chosen observables \cite{VQA,PQCReview2019}. Applied to graph-structured data, this template has been instantiated as quantum graph kernels, quantum graph neural networks, hybrid graph classifiers, and graph-generation pipelines \cite{QCE22Embedding,GraphQNTK,QCE24HQCGNN,QNetGAN}. Recent work has also used graph neural networks to analyze parameterized quantum circuits themselves \cite{QCE24GNN_EE}. These models, however, remain \textit{explicit finite-depth architectures}.

This matters because depth is a central design constraint for PQC-based models, and one that equilibrium methods are well-positioned to relieve. Deeper circuits are more expressive but accumulate hardware errors under realistic coherence budgets and are prone to \textit{barren plateaus}: training landscapes in which gradients vanish exponentially in the number of qubits or layers \cite{BarrenPlateauOriginal,Ragone2024,Fontana2024}. Each variational parameter further incurs a measurement cost through the parameter-shift rule used for gradient evaluation \cite{ParameterShift,QuantumNeuralNetworks}. These pressures favor shallower circuits with fewer independent parameters. 

Such a structure is precisely what an equilibrium approach can supply, by replacing an explicitly stacked PQC with a circuit whose output is defined as a fixed point. Recently introduced, Quantum Deep Equilibrium Models (QDEQs) \cite{QDEQ} realize this idea: a shallow circuit, evaluated by a root-finder, can match or exceed the accuracy of a much deeper explicit quantum network on unstructured data. However, application of this approach to the graph setting, the focus of our contribution here, has not been addressed.

\subsection{Contributions}\label{sec:intro-contrib}

We formulate three injection pathways that couple an encode--unitary--measure quantum module to a shared classical equilibrium map (classical backbone) whose fixed point would define the model state in the absence of any quantum component. The three pathways differ in where the quantum signal enters $\Phi_\Theta$ in the sense of Definition \ref{def:equilibrium}:

\begin{enumerate}
    \item[\textnormal{(i)}] \textit{Independent injection (ID)}: the quantum signal enters as an additive iterate-independent conditioning term, derived once from the encoder output $H$ and graph topology descriptors $ \tau(A)$, and held fixed across all solver steps.
    \item[\textnormal{(ii)}] \textit{State-dependent injection (SD)}: the quantum signal enters as a residual acting on the current iterate $Z^{(t)}$, and is recomputed at every solver step.
    \item[\textnormal{(iii)}] \textit{Backbone-dependent injection (BD)}: the same residual as in SD injection is applied instead to the output of the classical backbone at the current iterate.
\end{enumerate}

\noindent For each of these injection pathways we: (a) establish contraction guarantees under explicit assumptions on the Lipschitz constants of the classical backbone and quantum signal (Theorems \ref{theorem:independent-contraction}--\ref{theorem:backbone-contraction}), and (b) benchmark their training time and accuracy using a noiseless tensor network simulator against three TU Dortmund graph classification benchmarks (NCI1, PROTEINS, MUTAG), and compare to a classical IGNN baseline. We find that while SD and BD injections achieve accuracies similar to the classical IGNN baseline at the cost of higher training time, ID injection achieves improved accuracy for comparable training time.

\section{Related Work}\label{sec:related}

Our contribution sits at the intersection of three lines of research that have largely developed independently: classical implicit and equilibrium models (including their graph-specific refinements), quantum machine learning over graphs with PQCs, and quantum deep equilibrium models. The introduction sketched each; we now record the technical details most relevant to \S\ref{sec:method} and draw out the specific gap we address.

\subsection{Deep Equilibrium and Implicit Graph Models}

Beyond the basic DEQ construction \cite{DEQ,IDL}, stability and training of equilibrium models have received sustained attention. Certified \cite{CerDEQ2022} and Lyapunov-stable \cite{LyaDEQ2023} variants give formal guarantees on perturbation behavior, while the TorchDEQ library \cite{TorchDEQ} standardizes solver choices such as Picard, Anderson, and Broyden methods together with implicit-differentiation training options.

The IGNN equilibrium operator is bounded component-wise by a nonnegative matrix whose spectral radius is controlled through a constraint on the trainable weights, guaranteeing a unique fixed point \cite{IGNN}. Subsequent work refines this through monotone-operator theory, both in the general DEQ setting \cite{MonotoneIGNN} and for implicit graph models specifically \cite{Baker2023MonotoneIGNN}. Efficient and multiscale variants \cite{EIGNN,MGNNI} reduce solver cost, and alternative formulations are based on nonlinear diffusion dynamics \cite{GIND}. Our Theorems \ref{theorem:independent-contraction}--\ref{theorem:backbone-contraction} follow the contraction template of the IGNN analysis in Ref.~\cite{IGNN}, applied to an equilibrium operator carrying an additional quantum term.

\subsection{Quantum Graph Neural Networks}

Work on explicit quantum graph learning has developed four distinct strategies for where the quantum module enters the pipeline. At the \textit{kernel level}, graph data are processed classically and passed to a quantum-kernel classifier \cite{GraphQNTK}. At the \textit{encoder level}, node features are processed by a PQC before classical message passing \cite{QCE22Embedding}. At the \textit{message-passing level}, quantum circuits implement the per-layer update, either as learned convolutions \cite{QGCN,Hu2022QuGCNDesign} or as more general trainable quantum graph layers \cite{QGNN,GQGLA,InductiveQGNN}; Ref.~\cite{QGNNreview24} offers a survey of this strategy. At the \textit{readout level}, hybrid classifiers apply a quantum module after a classical graph representation is produced \cite{QCE24HQCGNN}. Separately, quantum graph generation for molecular design composes quantum modules into graph decoders \cite{QNetGAN}, and GNNs have been used to analyze parameterized quantum circuits rather than predict over graph data \cite{QCE24GNN_EE,HaQGNN}. None of these approaches transfers directly to the equilibrium operators characteristic of the implicit models we study here, but our pathways draw on two of them: ID injection operates at the encoder level, conditioning on the encoder output before the fixed-point solve, while SD and BD injection operate at the message-passing level, entering the per-iterate update of the equilibrium operator.

\subsection{Quantum Deep Equilibrium Models}

The QDEQ of Schleich \textit{et al.}~\cite{QDEQ} is the closest quantum precedent to the work we present here, being the first to combine equilibrium-style fixed-point training with parameterized quantum circuits. Two of its results directly parallel pieces of our own analysis: a universality theorem for weight-tied, input-injected quantum circuits, and a Lipschitz bound on the encode--unitary--measure map.

First, they extend the universality result of \cite{DEQ}: a weight-tied quantum circuit (layers sharing parameters $\theta$) with input injection is shown to be equivalent to an untied stack of independently parametrized quantum layers, in that the untied stack's output appears as one component of the tied network's equilibrium. This justifies the weight-tied, input-injected fixed point as a stand-in for arbitrarily deep quantum networks. Our architecture is weight-tied and input-injected in the same sense, allowing the quantum modules of our injection pathways to inherit the same universality.

Second, they establish Lipschitz continuity of the quantum encode--unitary--measure map itself: for amplitude or angle encoding and measurements restricted to Pauli operators or basis-state projectors, $\lVert f(\mathbf{z}) - f(\mathbf{z}')\rVert_2 \le c \|M\lVert_2\ket{\mathbf{z}}\bra{\mathbf{z}} - \ket{\mathbf{z'}}\bra{\mathbf{z}'}\rVert_\textnormal{Tr}$, where $c = 1$ for amplitude encoding and $c = 2$ for angle encoding, $\lVert M\rVert_2$ is the spectral norm of the measurement observable, and $f$ is quantum model function (see Definition 1 of Ref.~\cite{QDEQ}). This gives an analytic handle on the Lipschitz constants appearing in Theorems \ref{theorem:state-contraction}--\ref{theorem:backbone-contraction}: the quantum residual inherits this bound in composition with the linear maps that embed hidden states into and out of the quantum input space. Our setting differs in two respects: the input is graph-structured and carries topological information through $\tilde{A}$, and the equilibrium operator consists of a classical backbone injected with a quantum signal.

\subsection{Auxiliary Signal Coupling Strategies}

How an auxiliary signal is coupled to a host architecture has a long classical history, including feature-wise linear modulation (FiLM) \cite{FiLM} and residual-learning variants in hybrid quantum-classical settings \cite{ResHQCNN,ResQNets}. Related work on integration strategy effects in hybrid QML has examined embedding choice \cite{QCE22Embedding} and hyperparameter sensitivity \cite{QCE24OH_GNN}. These studies establish that coupling choice matters for robustness, expressivity, and trainability, but all operate within explicit, finite-depth forward passes. The operator-level question we study---\textit{where} an auxiliary signal enters a fixed-point equilibrium operator---only arises once the host architecture is a QDEQ, and does not appear in this literature. To the best of our knowledge, the three injection pathways we introduce are the first systematic taxonomy of this operator-level coupling choice for quantum modules in implicit graph models.

\begin{figure*}[!t]
    \centering
    \includegraphics[width=0.95\textwidth]{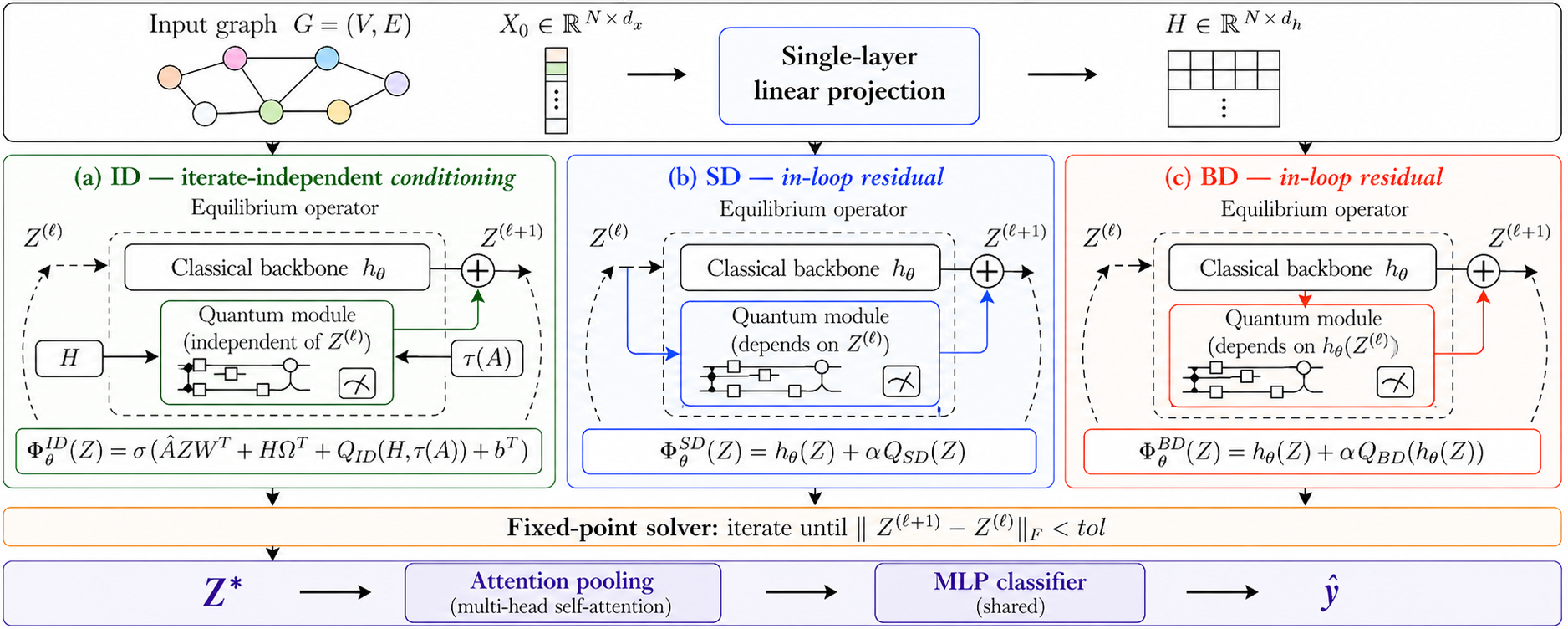}
    \caption{Quantum injection pathways in a shared implicit graph framework. Panels (a)--(c) show the ID, SD, and BD variants as three distinct equilibrium operators built atop the common classical backbone $h_\theta$ of Equation \eqref{eq:common-backbone}. The dashed box marks the fixed-point solver loop (Definition \ref{def:solver}): modules inside are re-evaluated at every solver step, while modules outside are computed once per forward propagation.}
    \label{fig:injection-pathways}
\end{figure*}

\section{Quantum Injection Pathways for IGNNs}\label{sec:method}

This section makes the injection pathway constructions sketched in \S\ref{sec:introduction} concrete. We first fix a common classical backbone---a shared encoder, propagation matrix, and classical equilibrium map---against which the three injection pathways can be compared on equal footing (\S\ref{sec:problem-setup}). We then specify the encode--unitary--measure \textit{quantum module} that the injection pathways share (\S\ref{sec:quantum-module}), and define the three coupling pathways as three distinct equilibrium operators built on top of the shared backbone (\S\ref{sec:pathways}). Finally, we establish well-posedness by deriving Lipschitz bounds on each injected equilibrium operator under explicit assumptions on the classical backbone and quantum modules (\S\ref{sec:contraction}).

\subsection{Problem Setup and Common Backbone}\label{sec:problem-setup}

Consider a graph $G=(V,E)$ with $N = |V|$ nodes, adjacency matrix $A \in \{0,1\}^{N \times N}$, and row-wise node features:

\begin{equation}
    X_0=[x_1^{(0)};\dots;x_N^{(0)}]\in\mathbb{R}^{N \times f}
\end{equation}

\noindent where each node $v\in V$ carries a feature vector $x_v^{(0)}\in\mathbb{R}^f$ and $f$ is the input feature dimension. A common encoder maps this raw representation to hidden node features:

\begin{equation}
    H = \mathrm{Enc}(A, X_0) = [h_1;\dots;h_N] \in \mathbb{R}^{N \times d_h}
\end{equation}

\noindent where $d_h$ is the hidden feature dimension, a model hyperparameter fixed across all pathways.

Let $\tilde{A} \in \mathbb{R}^{N \times N}$ denote the \textit{normalized propagation matrix}. In our experiments, we use the symmetrically normalized adjacency with self-loops introduced by Kipf and Welling \cite{GCN}:

\begin{equation}\label{eq:normalized-adjacency}
    \tilde{A} = \hat{D}^{-1/2}\,(A+I_N)\,\hat{D}^{-1/2}
\end{equation}

\noindent where $I_N$ is the $N\times N$ identity and $\hat{D}$ is the diagonal degree matrix of $A+I_N$. This symmetric normalization rescales contributions from each neighbor by the geometric mean of the two endpoints' degrees and ensures $\|\tilde{A}\|_2 \le 1$, which will feature directly in the contraction bounds of \S\ref{sec:contraction}.

Writing $\Theta$ for the full pathway parameters---with $\theta$ collecting the classical backbone weights and $\eta$ a pathway-specific collection of quantum parameters---all three variants define an equilibrium state (Definition \ref{def:equilibrium}):

\begin{equation}\label{eq:common-equilibrium}
    Z^\star = \Phi_{\Theta}(Z^\star;\tilde{A}, H)
\end{equation}

\noindent under the same encoder, propagation matrix, and hidden-state space. This shared-backbone setup lets us compare quantum injection pathways on equal footing.

\begin{definition}[Classical backbone]\label{def:backbone}
    The \textit{classical backbone} associated to parameters $\theta=\{W,\Omega,b\}$, propagation matrix $\tilde A$, and encoder output $H$ is the map:
    
    \begin{equation}\label{eq:common-backbone}
        h_\theta(Z) = \sigma\left(\tilde{A}\, Z\, W^\top + H\, \Omega^\top + \one\, b^\top\right)
    \end{equation}
    
    \noindent where $W,\Omega \in \mathbb{R}^{d_h \times d_h}$ are trainable matrices, $b \in \mathbb{R}^{d_h}$ is a shared bias, $\one\in\mathbb{R}^N$ denotes the all-ones column vector, and $\sigma(\cdot)=\tanh(\cdot)$. By construction, $h_\theta$ is the equilibrium operator (Definition \ref{def:equilibrium}) of an underlying classical IGNN.
\end{definition}

Following the equilibrium-learning viewpoint of Refs.~\cite{DEQ,IGNN}, we use the shared classical backbone, defined in Definition \ref{def:backbone}. The three terms inside $\sigma$ play distinct roles: $\tilde{A} Z W^\top$ is the message-passing step, aggregating each node's neighbors via $\tilde{A}$ and mixing hidden features via $W$; $H \Omega^\top$ is the input-injection term referenced in the introduction, which carries the encoder output into the fixed-point equation as an iterate-independent bias; and $\one b^\top$ is a row-broadcast bias.

\subsection{Quantum Module Definition}\label{sec:quantum-module}

All three injection pathways instantiate the encode--unitary--measure pattern at the node level \cite{QDEQ,PQCReview2019}: a classical input is loaded into a quantum state by a data-dependent unitary, processed by a trainable circuit, and read out through expectation values of chosen observables. We describe this shared construction once here, parameterized by a generic input/output dimension $(d_\mathrm{in}, d_\mathrm{out})$ and trainable parameters $\eta$; the three injection pathways differ only in their input, their parameters, and where the output enters the equilibrium operator. These details are specified in \S\ref{sec:pathways}. 

Each quantum module $q_\eta:\mathbb{R}^{d_\mathrm{in}}\to\mathbb{R}^{d_\mathrm{out}}$ maps a node-specific input $s_v\in\mathbb{R}^{d_\mathrm{in}}$ to an output in three stages. A fourth step that lifts the node-specific outputs to an aggregated all-node matrix yields the state which is then added to the classical backbone in an injection pathway-specific manner.

\subsubsection{Classical-to-Quantum Interface} 

The quantum circuit operates on $n_q$ qubits (per node $v\in V$) and therefore accepts inputs in $\mathbb{R}^{n_q}$; note that $n_q\neq d_\mathrm{in}$ in general. A learned linear map $W_{\mathrm{in}}\in\mathbb{R}^{n_q\times d_\mathrm{in}}$ is therefore employed per node to map hidden input features of dimension $d_\mathrm{in}$ to a size amenable to the quantum circuit, followed by a componentwise $\tanh(\cdot)$:

\begin{equation}\label{eq:qmod-bottleneck}
    u_v = \tanh(W_{\mathrm{in}}\,s_v) \in (-1,1)^{n_q}
\end{equation}

\noindent The $\tanh(\cdot)$ nonlinearity bounds each entry in $(-1,1)$---a requirement of the angle-encoding unitary below.

\subsubsection{Data Encoding, Parametrized Circuit, and Measurement} 

The bounded feature $u_v$ is loaded into a pure $n_q$-qubit state via an angle-encoding unitary $S_{u_v}$:

\begin{equation}\label{eq:qmod-encoding}
    \rho^{\mathrm{enc}} = S_{u_v}\ket{0}^{\otimes n_q}\bra{0}^{\otimes n_q}S_{u_v}^\dagger
\end{equation}

\noindent where each component $u_v^{(j)}$ sets the rotation angle of a single-qubit gate on qubit $j$ \cite{QDEQ}. A trainable $n_q$-qubit unitary $U(\eta)$ is then applied, and for each qubit $j$ the expectation value of a Hermitian observable $M_j$ is computed:

\begin{equation}\label{eq:qmod-measurement}
    m_v^{(j)} = \mathrm{Tr}\bigl[M_j\,U(\eta)\,\rho^{\mathrm{enc}}\,U(\eta)^\dagger\bigr],\quad j=1,\dots,n_q
\end{equation}

\noindent yielding a vector $m_v=[m_v^{(1)},\dots,m_v^{(n_q)}]^\top\in\mathbb{R}^{n_q}$. Each $M_j$ is taken to be a single-qubit Pauli operator ($\|M_j\|_2=1$, expectations in $[-1,1]$), and $U(\eta)$ is instantiated as the repeated \textit{deep XYZ ansatz} described in Appendix \ref{app:ansatz}, inspired by the data re-uploading design of Ref.~\cite{DRQNN}.

\subsubsection{Quantum-to-Classical Interface} 

The expectation vector $m_v$ is lifted back to $\mathbb{R}^{d_\mathrm{out}}$ via a linear map $W_{\mathrm{out}}\in\mathbb{R}^{d_\mathrm{out}\times n_q}$ such that $q_\eta(s_v)=W_{\mathrm{out}}m_v$. In the implementation, $W_{\mathrm{in}}$ and $W_{\mathrm{out}}$ are ordinary trainable linear layers learned jointly with the rest of the model by back-propagation; for the in-loop SD/BD residual modules, the released code additionally applies spectral normalization to these two maps to help control the residual Lipschitz budget. Applied independently to each row of a matrix-valued state $S=[s_1;\dots;s_N]$ with shared parameters $(\eta,W_{\mathrm{in}},W_{\mathrm{out}})$, the node-wise map $q_\eta(s_v)$ induces the row-wise map:

\begin{equation}\label{eq:qmod-rowwise}
    Q(S) := \bigl[q_\eta(s_1);\dots;q_\eta(s_N)\bigr]
\end{equation}

\noindent where $Q\in\mathbb{R}^{N\times d_\mathrm{out}}$. This matrix is then added to the classical backbone according to one of the three injection pathways.

\subsection{The Three Injection Pathways}\label{sec:pathways}

Each pathway defines a distinct equilibrium operator in the sense of Definition~\ref{def:equilibrium} by coupling the quantum module of \S\ref{sec:quantum-module} to the shared classical backbone $h_\theta$ of Equation \eqref{eq:common-backbone}. The three differ along two axes: (i) whether the quantum module is evaluated \textit{once per forward solve} or \textit{at every solver step} (in the sense of Definition~\ref{def:solver}), and (ii) when evaluated in-loop, what the module's input is. We refer to independent injection as an \textit{iterate-independent conditioning} pathway, and to the remaining two as \textit{in-loop residual} pathways. Figure~\ref{fig:injection-pathways} summarizes the resulting operator designs.

\subsubsection{Independent Injection}

Independent injection evaluates the quantum module once per forward solve, \textit{before} the fixed-point iteration begins, and holds its output fixed throughout. The node-level input is the encoder feature $h_i$ augmented with deterministic node-level topology descriptors:

\begin{equation}
    \tau_i(A)=\bigl[c_i^{(3)},\dots,c_i^{(L_{\max})},\,d_i,\,\mathrm{cc}_i,\,\mathbf{1}_{\triangle,i}\bigr]\in\mathbb{R}^{d_\tau} 
\end{equation}

\noindent where $d_\tau=(L_{\max}-2)+3$, $c_i^{(\ell)}$ counts cycles of length $\ell$ containing node $i$ in the cycle basis used by the implementation, $d_i$ is the normalized degree, $\mathrm{cc}_i$ is the clustering coefficient, and $\mathbf{1}_{\triangle,i}$ indicates triangle membership. These descriptors are deterministic functions of $A$ and are not learned.

The quantum module of \S\ref{sec:quantum-module} is instantiated with trainable parameters $\eta$ and node-level input $[h_i;\tau_i(A)]\in\mathbb{R}^{d_\mathrm{in}+d_\tau}$; we write this module as $q_\eta^{\mathrm{ID}}$. The only change relative to \S\ref{sec:quantum-module} is that the classical-to-quantum bottleneck $W_{\mathrm{in}}$ of Equation \eqref{eq:qmod-bottleneck} has shape $n_q\times(d_h+d_\tau)$ to accommodate the augmented input; the encoding \eqref{eq:qmod-encoding}, measurement \eqref{eq:qmod-measurement}, and quantum-to-classical stages are unchanged. Applying $q_\eta^\mathrm{ID}$ row-wise yields:

\begin{equation}\label{eq:qind-def}
    Q_{\mathrm{ID}}(H;\tau(A)):=\bigl[q_\eta^\mathrm{ID}(h_1;\tau_1(A));\dots;q_\eta^\mathrm{ID}(h_N;\tau_N(A))\bigr]
\end{equation}

\noindent where $Q_\mathrm{ID}\in\mathbb{R}^{N\times d_h}$. This is added to $\Phi_\Theta$ as an iterate-independent conditioning term:

\begin{equation}\label{eq:independent-injection}
    \Phi_{\Theta}^{\mathrm{ID}}(Z)=\sigma\!\left(\tilde{A}\,Z\,W^\top + H\,\Omega^\top + Q_{\mathrm{ID}}(H,\tau(A)) + \mathbf{1}\,b^\top\right)
\end{equation}

\noindent With the graph and its encoder output fixed, $Q_{\mathrm{ID}}(H,\tau(A))$ does not depend on $Z$; it sits parallel to the linear input-injection term $H\Omega^\top$ as a precomputed bias that the fixed-point solver sees as a constant. Consequently, independent injection introduces no additional $\partial Q/\partial Z$ contribution to the fixed-point derivative.

\subsubsection{State-Dependent Injection}

State-dependent injection evaluates the quantum module \textit{inside} the fixed-point recurrence with the current iterate $Z$ as its input. Writing:

\begin{equation}
    Q_\mathrm{SD}(Z):=\bigl[q_\eta^\mathrm{SD}(z_1);\dots;q_\eta^\mathrm{SD}(z_N)\bigr]\in\mathbb{R}^{N\times d_h}
\end{equation}

\noindent for the row-wise map of Equation \eqref{eq:qmod-rowwise} with pathway-specific parameters $\eta$, the equilibrium operator is:

\begin{equation}\label{eq:state-dependent}
    \Phi_{\Theta}^{\mathrm{SD}}(Z)=h_\theta(Z) + \alpha\, Q_\mathrm{SD}(Z)
\end{equation}

\noindent where $\alpha> 0$ is a residual scale (a model hyperparameter). The quantum residual $Q_\mathrm{SD}(Z)$ is recomputed at every fixed point solver step because its input $Z$ is updated at every step.

\subsubsection{Backbone-Dependent Injection}

Backbone-dependent injection, like state-dependent injection, evaluates the quantum module \textit{inside} the fixed-point recurrence, but takes the classical backbone output of $h_\theta(Z)$ as its input. Writing:

\begin{equation}
    Q_\mathrm{BD}(Z):=\bigl[q_\eta^\mathrm{BD}(z_1);\dots;q_\eta^\mathrm{BD}(z_N)\bigr]\in\mathbb{R}^{N\times d_h}
\end{equation}

\noindent for the row-wise map of Equation \eqref{eq:qmod-rowwise} with pathway-specific parameters $\eta$, the equilibrium operator is:

\begin{equation}\label{eq:backbone-dependent}
    \Phi_{\Theta}^{\mathrm{BD}}(Z)=h_\theta(Z) + \alpha\, Q_\mathrm{BD}\!\bigl(h_\theta(Z)\bigr)
\end{equation}

\noindent where $\alpha\ge 0$ is a residual scale. The quantum residual $Q_\mathrm{BD}(Z)$ is recomputed at every fixed point solver step because its input $Z$ is updated at every step. Note that Equations \eqref{eq:state-dependent} and \eqref{eq:backbone-dependent} differ only in the argument of $Q_{\bullet}$. This single structural change propagates through the Lipschitz analysis of \S\ref{sec:contraction}, yielding strictly different contraction budgets for the two pathways despite identical quantum machinery.

\medskip
The three operators thus span a natural range of couplings: ID keeps the quantum signal outside the state-dependent part of the recurrence, contributes no additional $Z$-dependent Jacobian term; BD places it inside the recurrence but after the classical contraction; and SD places it in direct contact with the iterate. \S\ref{sec:contraction} makes this intuition precise, bounding $\mathrm{Lip}_F(\Phi_\Theta^{\bullet})$ for each pathway and ordering the resulting contraction budgets.

\subsection{Well-posedness and Lipschitz Contraction Analysis}\label{sec:contraction}

We now establish sufficient conditions under which each pathway's equilibrium operator is a contraction---and therefore admits a unique fixed point by Banach's fixed-point theorem \cite{Rudin1976}. The three results (Theorems \ref{theorem:independent-contraction}--\ref{theorem:backbone-contraction}) follow the template of the IGNN analysis in Ref.~\cite{IGNN}, applied to operators that carry an additional quantum term.

Throughout the analysis, matrix-valued state perturbations are measured in the Frobenius norm $\|\cdot\|_F$, and fixed matrices acting by left or right multiplication are bounded in the induced spectral norm $\|\cdot\|_2$. The two norms interact through submultiplicativity. For compatible matrices $A,X,B$:

\begin{equation}\label{eq:norm-submult}
    \|AXB\|_F \le \|A\|_2 \|X\|_F \|B\|_2
\end{equation}

\noindent and $\|M\|_2 \le \|M\|_F$ for every matrix $M$ \cite{HornJohnson2012}. A map $T$ acting on matrix-valued states is said to be $L$-Lipschitz in the Frobenius norm---written $\mathrm{Lip}_F(T) \le L$---if $\|T(S_1)-T(S_2)\|_F \le L\|S_1-S_2\|_F$ for all $S_1,S_2$ in its domain; a contraction is a map with $\mathrm{Lip}_F(T) < 1$. Since $\sigma(\cdot)=\tanh(\cdot)$ is component-wise non-expansive, the induced map $\sigma:\mathbb{R}^{N\times d_h}\to\mathbb{R}^{N\times d_h}$ satisfies $\|\sigma(P)-\sigma(Q)\|_F \le \|P-Q\|_F$ for all $P,Q$, a fact we use repeatedly below.

First, we derive the Lipschitz constant of the classical backbone (Lemma \ref{lem:backbone-lip}). Then we derive the quantum module Lipschitz constants that enter the injection pathway-specific contraction theorems by bounding the node-level quantum map $q_\eta$ of \S\ref{sec:quantum-module} (Lemma \ref{lem:nodewise-lip}), and lift this bound to the row-wise matrix-valued maps $Q_\mathrm{ID}$, $Q_\mathrm{SD}$, $Q_\mathrm{BD}$ defined in \S\ref{sec:pathways} (Lemma \ref{lem:rowwise-lip}). Following this, we derive the Lipschitz constants for each \textit{full} quantum injection pathway and state the conditions required for contraction.

\begin{lemma}[Contraction of the classical backbone]\label{lem:backbone-lip}
    Let $h_\theta$ denote the classical backbone of Equation \eqref{eq:common-backbone}. If $\sigma$ is component-wise non-expansive and $\|W\|_2\le\kappa/\|\tilde{A}\|_2$ for some $0\le\kappa<1$, then for all $Z_1,Z_2\in\mathbb{R}^{N\times d_h}$:
    
    \begin{equation}\label{eq:backbone-lip}
        \|h_\theta(Z_1)-h_\theta(Z_2)\|_F\le\kappa\,\|Z_1-Z_2\|_F
    \end{equation}
\end{lemma}

\begin{proof}
    Writing $C := H\Omega^\top + \one b^\top$, which is constant in $Z$:
    
    \begin{equation}
        h_\theta(Z_1) - h_\theta(Z_2) = \sigma\!\bigl(\tilde{A}Z_1 W^\top + C\bigr) - \sigma\!\bigl(\tilde{A}Z_2 W^\top + C\bigr)
    \end{equation}
    
    \noindent Since $\sigma$ is component-wise non-expansive and $C$ cancels in the difference, applying Equation \eqref{eq:norm-submult} gives $\|h_\theta(Z_1)-h_\theta(Z_2)\|_F \le \|\tilde{A}\|_2\|W\|_2\|Z_1-Z_2\|_F \le \kappa\,\|Z_1-Z_2\|_F$.
\end{proof}

\begin{lemma}[Node-level Lipschitz bound on the quantum map]\label{lem:nodewise-lip}
    Let $q_{\eta}:\mathbb{R}^{d_\mathrm{in}}\to\mathbb{R}^{d_\mathrm{out}}$ denote the node-level quantum map of \S\ref{sec:quantum-module}, given by the composition of the classical-to-quantum interface (Equation \eqref{eq:qmod-bottleneck}), the angle-encoding unitary (Equation \eqref{eq:qmod-encoding}), the PQC and measurement stage (Equation \eqref{eq:qmod-measurement}), and the quantum-to-classical interface(Equation \eqref{eq:qmod-rowwise}). If the measurement observables satisfy $\|M_j\|_2\le 1$, then $q_{\eta}$ is Lipschitz in the Euclidean norm with constant:
    
    \begin{equation}\label{eq:qmod-lipschitz}
        L_q \;\le\; 2\sqrt{n_q}\,\|W_\mathrm{out}\|_2\,\|W_\mathrm{in}\|_2
    \end{equation}
\end{lemma}

\begin{proof}
    Decompose $q_{\eta}=W_\mathrm{out}\circ f_M\circ\tanh\circ W_\mathrm{in}$, where $f_M:u\mapsto (m^{(1)}(u),\dots,m^{(n_q)}(u))^\top\in\mathbb{R}^{n_q}$ collects the Pauli expectation values of Equation \eqref{eq:qmod-measurement}. Fix $s,s'\in\mathbb{R}^{d_\mathrm{in}}$ and set $u=\tanh(W_\mathrm{in}s)$, $u'=\tanh(W_\mathrm{in}s')$.

    \smallskip
    
    \noindent \textbf{(i) Classical-to-quantum interface.} Component-wise non-expansiveness of $\tanh(\cdot)$ together with submultiplicativity of the spectral norm gives:
    
    \begin{equation}
        \|u-u'\|_2
        \le \|W_\mathrm{in}(s-s')\|_2
        \le \|W_\mathrm{in}\|_2\,\|s-s'\|_2
    \end{equation}
    
    \noindent \textbf{(ii) Encode--unitary--measure.} By Lemma 1 of \cite{QDEQ}, for each angle-encoded measurement component with Hermitian observable $M_j$:
    
    \begin{equation}
        |m^{(j)}(u)-m^{(j)}(u')|
        \le 2\|M_j\|_2\,\|u-u'\|_2
        \le 2\|u-u'\|_2
    \end{equation}
    
    \noindent where the second inequality uses $\|M_j\|_2\le 1$. Summing over the $n_q$ measured components gives:
    
    \begin{align}
        \|f_M(u)-f_M(u')\|_2^2
        &= \sum_{j=1}^{n_q} |m^{(j)}(u)-m^{(j)}(u')|^2 \nonumber \\
        &\le 4 n_q \|u-u'\|_2^2
    \end{align}
    
    \noindent Therefore:
    
    \begin{equation}
        \|f_M(u)-f_M(u')\|_2
        \le 2\sqrt{n_q}\,\|u-u'\|_2
    \end{equation}
    
    \noindent \textbf{(iii) Quantum-to-classical interface.} Submultiplicativity applied to $W_\mathrm{out}$ gives:
    
    \begin{align}
        \|W_\mathrm{out}\bigl(f_M(u)-f_M(u')\bigr)\|_2
        &\le \|W_\mathrm{out}\|_2\,\|f_M(u)-f_M(u')\|_2
    \end{align}
    
    \noindent Chaining (i)--(iii) yields:
    
    \begin{equation}
        \|q_{\eta}(s)-q_{\eta}(s')\|_2
        \le
        2\sqrt{n_q}\,
        \|W_\mathrm{out}\|_2
        \|W_\mathrm{in}\|_2
        \|s-s'\|_2
    \end{equation}
\end{proof}

\begin{lemma}[Lipschitz extension to the row-wise map]\label{lem:rowwise-lip}
    Let $Q:\mathbb{R}^{N\times d_\mathrm{in}}\to\mathbb{R}^{N\times d_\mathrm{out}}$ be the row-wise map:
    
    \begin{equation}
        Q(S):=\bigl[q_{\eta}(s_1);\dots;q_{\eta}(s_N)\bigr]
    \end{equation}
    
    \noindent associated to the node-level map $q_{\eta}$ as in Equation \eqref{eq:qmod-rowwise}. If $q_{\eta}$ is $L_q$-Lipschitz in the $\ell_2$ norm, then $\mathrm{Lip}_F(Q)\le L_q$.
\end{lemma}

\begin{proof}

    Identifying each row vector $s\in\mathbb{R}^{d_\mathrm{in}}$ with a $1\times d_\mathrm{in}$ matrix so that $\|s\|_2=\|s\|_F$, we have for $S_1=[s_1^{(1)};\dots;s_N^{(1)}]$ and $S_2=[s_1^{(2)};\dots;s_N^{(2)}]$:
    
    \begin{align}
        \|Q(S_1)-Q(S_2)\|_F^2 
        &= \sum_{i=1}^N \|q_{\eta}(s_i^{(1)})-q_{\eta}(s_i^{(2)})\|_2^2 \nonumber \\
        &\le L_q^2 \sum_{i=1}^N \|s_i^{(1)}-s_i^{(2)}\|_2^2 \nonumber \\
        &= L_q^2\,\|S_1-S_2\|_F^2
    \end{align}
    
    \noindent where the first equality uses the block-diagonal structure of the row-wise map and the inequality uses the node-level $L_q$-Lipschitz hypothesis on $q_{\eta}$. Taking square roots gives $\mathrm{Lip}_F(Q)\le L_q$.
\end{proof}

Lemmas \ref{lem:backbone-lip}--\ref{lem:rowwise-lip} apply to each of the three row-wise maps $Q_\mathrm{ID}$, $Q_\mathrm{SD}$, $Q_\mathrm{BD}$ of \S\ref{sec:pathways}: all three share the encode--unitary--measure structure of \S\ref{sec:quantum-module}, differing only in their input dimension and in their pathway-specific parameters $(\eta,W_\mathrm{in},W_\mathrm{out})$. In particular, $Q_\mathrm{ID}$ has input dimension $d_\mathrm{in}=d_h+d_\tau$, while $Q_\mathrm{SD}$ and $Q_\mathrm{BD}$ have input dimension $d_\mathrm{in}=d_h$. We therefore write $L_q^\mathrm{ID}$, $L_q^\mathrm{SD}$, $L_q^\mathrm{BD}$ for the corresponding Lipschitz constants, each satisfying Equation \eqref{eq:qmod-lipschitz} evaluated per pathway. We now prove the Lipschitz constants and contraction requirements for each injection pathway.

\begin{theorem}[Contraction of independent injection]\label{theorem:independent-contraction}
    Assume $\sigma$ is component-wise non-expansive and $\|W\|_2\le\kappa/\|\tilde{A}\|_2$ for some $0\le\kappa<1$. Then, for a fixed graph, fixed encoder output $H$, fixed topology descriptors $\tau(A)$, and the induced iterate-independent term $Q_\mathrm{ID}(H,\tau(A))$, the ID operator of Equation \eqref{eq:independent-injection} satisfies:
    
    \begin{equation}
        \mathrm{Lip}_F(\Phi_{\Theta}^{\mathrm{ID}})\le\kappa
    \end{equation}
    
    \noindent Consequently, $\Phi_{\Theta}^{\mathrm{ID}}$ is a contraction on $(\mathbb{R}^{N\times d_h},\|\cdot\|_F)$ and admits a unique equilibrium to which fixed-point iteration converges \cite{IGNN,Rudin1976}.
\end{theorem}

\begin{proof}
    By the same argument as Lemma \ref{lem:backbone-lip}, with $\widetilde C := H\Omega^\top + Q_\mathrm{ID}(H,\tau(A)) + \one b^\top$ in place of $C$ (which is still constant in $Z$ for fixed graph and encoder output), we obtain $\|\Phi_{\Theta}^{\mathrm{ID}}(Z_1)-\Phi_{\Theta}^{\mathrm{ID}}(Z_2)\|_F \le \kappa\,\|Z_1-Z_2\|_F$. The quantum term $Q_\mathrm{ID}(H,\tau(A))$ does not depend on $Z$ and therefore cancels inside $\sigma$; in particular, $L_q^\mathrm{ID}$ does not appear in the bound.
\end{proof}

The key observation is that $Q_\mathrm{ID}(H,\tau(A))$, being constant in $Z$, cancels when $\Phi_\Theta^{\mathrm{ID}}(Z_1)-\Phi_\Theta^{\mathrm{ID}}(Z_2)$ is formed, so the quantum term contributes nothing to the Lipschitz bound; in particular, $L_q^\mathrm{ID}$ does not appear on the right-hand side.

\begin{theorem}[Contraction of state-dependent injection]\label{theorem:state-contraction}
    Under the same assumptions on $\sigma$ and $\|W\|_2$ as in Theorem \ref{theorem:independent-contraction}, let $L_q^\mathrm{SD}$ denote the Lipschitz constant of the row-wise quantum map $Q_\mathrm{SD}:\mathbb{R}^{N\times d_h}\to\mathbb{R}^{N\times d_h}$ supplied by Lemmas \ref{lem:nodewise-lip} and \ref{lem:rowwise-lip}. Then the SD operator of Equation \eqref{eq:state-dependent} satisfies:
    
    \begin{equation}\label{eq:sd-lipschitz}
        \mathrm{Lip}_F(\Phi_{\Theta}^{\mathrm{SD}})\le\kappa+\alpha\,L_q^\mathrm{SD}
    \end{equation}
    
    \noindent If $\kappa+\alpha\,L_q^\mathrm{SD}<1$, then $\Phi_{\Theta}^{\mathrm{SD}}$ is a contraction on $(\mathbb{R}^{N\times d_h},\|\cdot\|_F)$ and admits a unique equilibrium to which fixed-point iteration converges \cite{IGNN,Rudin1976}.
\end{theorem}

\begin{proof}
    For any $Z_1,Z_2\in\mathbb{R}^{N\times d_h}$, the triangle inequality gives:
    
    \begin{equation}
        \begin{aligned}
            \|\Phi_{\Theta}^{\mathrm{SD}}(Z_1)-\Phi_{\Theta}^{\mathrm{SD}}(Z_2)\|_F 
            &\le \|h_\theta(Z_1)-h_\theta(Z_2)\|_F \\
            &\quad + \alpha\,\|Q_\mathrm{SD}(Z_1)-Q_\mathrm{SD}(Z_2)\|_F \\
            &{\le} (\kappa+\alpha\,L_q^\mathrm{SD})\,\|Z_1-Z_2\|_F
        \end{aligned}
    \end{equation}
    
    \noindent where the second inequality applies Lemma \ref{lem:backbone-lip}, and $\mathrm{Lip}_F(Q_\mathrm{SD})\le L_q^\mathrm{SD}$ from Lemma \ref{lem:rowwise-lip} on the quantum term.
\end{proof}

Unlike the case of independent injection, the quantum term does not cancel in the difference $\Phi_\Theta^{\mathrm{SD}}(Z_1)-\Phi_\Theta^{\mathrm{SD}}(Z_2)$ because $Q_\mathrm{SD}$ acts directly on the evolving state; the Lipschitz budget must therefore absorb $\alpha\,L_q^\mathrm{SD}$ in full.

\begin{theorem}[Contraction of backbone-dependent injection]\label{theorem:backbone-contraction}
    Under the same assumptions on $\sigma$ and $\|W\|_2$ as in Theorem \ref{theorem:independent-contraction}, let $L_q^\mathrm{BD}$ denote the Lipschitz constant of the row-wise quantum map $Q_\mathrm{BD}:\mathbb{R}^{N\times d_h}\to\mathbb{R}^{N\times d_h}$ supplied by Lemmas \ref{lem:nodewise-lip} and \ref{lem:rowwise-lip}. Then the BD operator of Equation \eqref{eq:backbone-dependent} satisfies:
    
    \begin{equation}\label{eq:bd-lipschitz}
        \mathrm{Lip}_F(\Phi_{\Theta}^{\mathrm{BD}})\le\kappa(1+\alpha\,L_q^\mathrm{BD})
    \end{equation}
    
    \noindent If $\kappa(1+\alpha\,L_q^\mathrm{BD})<1$, then $\Phi_{\Theta}^{\mathrm{BD}}$ is a contraction on $(\mathbb{R}^{N\times d_h},\|\cdot\|_F)$ and admits a unique equilibrium to which fixed-point iteration converges \cite{IGNN,Rudin1976}.
\end{theorem}

\begin{proof}

    By the triangle inequality:
    
    \begin{equation}\label{eq:bd-triangle}
        \begin{aligned}
        \|\Phi_{\Theta}^{\mathrm{BD}}(Z_1)&-\Phi_{\Theta}^{\mathrm{BD}}(Z_2)\|_F 
        \nonumber \\ &\le \|h_\theta(Z_1)-h_\theta(Z_2)\|_F \\
        &\qquad + \alpha\,\|Q_\mathrm{BD}(h_\theta(Z_1))-Q_\mathrm{BD}(h_\theta(Z_2))\|_F
        \end{aligned}
    \end{equation}
    
    \noindent The first term is bounded as in Theorem \ref{theorem:independent-contraction}:
    
    \begin{equation}\label{eq:bd-back}
        \|h_\theta(Z_1)-h_\theta(Z_2)\|_F\le\kappa\,\|Z_1-Z_2\|_F
    \end{equation}
    
    \noindent For the second, Lemma \ref{lem:rowwise-lip} applied at the points $h_\theta(Z_1),h_\theta(Z_2)$ composed with Equation \eqref{eq:bd-back} gives:
    
    \begin{align}\label{eq:bd-qterm}
        \|Q_\mathrm{BD}(h_\theta(Z_1))&-Q_\mathrm{BD}(h_\theta(Z_2))\|_F \nonumber \\
        &\le L_q^\mathrm{BD}\,\|h_\theta(Z_1)-h_\theta(Z_2)\|_F \nonumber \\
        &\le L_q^\mathrm{BD}\,\kappa\,\|Z_1-Z_2\|_F
    \end{align}
    
    \noindent Substituting Equations \eqref{eq:bd-back} and \eqref{eq:bd-qterm} into Equation \eqref{eq:bd-triangle}:
    
    \begin{equation}
        \|\Phi_{\Theta}^{\mathrm{BD}}(Z_1)-\Phi_{\Theta}^{\mathrm{BD}}(Z_2)\|_F \le \kappa(1+\alpha\,L_q^\mathrm{BD})\,\|Z_1-Z_2\|_F
    \end{equation}
    
    \noindent where the second inequality applies Lemma \ref{lem:backbone-lip}, and $\mathrm{Lip}_F(Q_\mathrm{BD})\le L_q^\mathrm{BD}$ from Lemma \ref{lem:rowwise-lip} on the quantum term.
\end{proof}

The factor $\kappa$ multiplying $\alpha\,L_q^\mathrm{BD}$ reflects that the quantum residual acts on the backbone output $h_\theta(Z)$, which is itself $\kappa$-Lipschitz in $Z$; perturbations in $Z$ are therefore attenuated by the classical contraction before reaching the quantum stage.

\begin{remark}[Ordering of the Lipschitz bounds]\label{rem:ordering}
    At a matched Lipschitz budget $L_q:=L_q^\mathrm{SD}=L_q^\mathrm{BD}$---corresponding to the two dependent pathways sharing the same circuit structure and spectral norms $\|W_\mathrm{in}\|_2,\|W_\mathrm{out}\|_2$---the three bounds satisfy $\kappa\le\kappa(1+\alpha L_q)\le\kappa+\alpha L_q$ for $0\le\kappa\le 1$ and $\alpha L_q\ge 0$. Independent injection therefore admits contraction under the weakest hypothesis (just $\kappa<1$), backbone-dependent injection under an intermediate hypothesis ($\kappa(1+\alpha L_q)<1$), and state-dependent injection under the strongest ($\kappa+\alpha L_q<1$).
\end{remark}

\subsection{Readout and Classification Head}\label{sec:readout}

Once the equilibrium state $Z^\star$ is computed by solving Equation \eqref{eq:common-equilibrium}, the model forms a fixed-size graph representation through a permutation-invariant \textit{readout} operation that aggregates across nodes:

\begin{equation}\label{eq:readout}
    z_G = \mathrm{READOUT}(Z^\star)\in\mathbb{R}^{d_h}
\end{equation}

\noindent In our experiments, $\mathrm{READOUT}$ is attention pooling. Concretely, a multi-head self-attention \cite{Transformer} block is applied to the node states within each graph, and the attended node representations are then summed to produce $z_G$. A multilayer perceptron (MLP) classifier head maps $z_G$ to the prediction:

\begin{equation}\label{eq:classifier}
    \hat{y}=\mathrm{MLP}(z_G)
\end{equation}

\noindent Both the readout and the classifier head are shared across all three injections and are not involved in the equilibrium solve. Their parameters are trained jointly with $\Theta$ via gradients back-propagated through the implicit-differentiation step.

\section{Experiments}\label{sec:experiments}

The experiments in this section serve two purposes. First, they test whether the three injection pathways defined in \S\ref{sec:method}, despite sharing the classical backbone of Definition \ref{def:backbone}, separate empirically along test accuracy and forward-solver cost. Second, they probe whether any observed separation is consistent with the matched-budget ordering of the contraction bounds (Remark \ref{rem:ordering}). We accordingly report both predictive accuracy and the average number of forward solver iterations per fixed-point solve as primary observables, alongside training time as a practical cost. We do not claim that our contraction bounds are tight, only that they offer a qualitative explanation for the ordering observed in solver behavior.

\subsection{Experimental Setup}

We evaluate the three quantum injection pathways on the TU graph-classification benchmarks NCI1, PROTEINS, and MUTAG \cite{TUDatasets}. Unless otherwise stated, evaluation follows fold-based TU protocols with multiple random seeds, in line with fair and reproducible GNN benchmarking practice \cite{FairComparisonGNNGraphClass,BenchmarkingGNNs}. We report test accuracy, forward DEQ solver iterations, and training time.

The main comparison includes a classical equilibrium baseline and the independent, state-dependent, and backbone-dependent quantum-injection variants. The classical baseline uses the same implicit graph backbone as the proposed models but removes the quantum pathway, isolating how the auxiliary signal enters the fixed-point map. For Table \ref{tab:pathway-main}, results are aggregated over three random seeds ($42$, $123$, $456$) and $10$ folds, so the reported mean $\pm$ standard deviation are over fold--seed runs. The ``Iter.'' column reports the mean number of \emph{forward} DEQ solver iterations at the final training epoch, and ``Time (m)'' reports wall-clock training time per run in minutes, both averaged over the reported runs.

All four variants share the same classical backbone and training setup. The encoder is a single linear projection to hidden dimension $64$ (\texttt{min\_encoder}), and the classifier is a two-layer MLP, $64 \rightarrow 64 \rightarrow C$, with ReLU and dropout $0.4$. Optimization uses AdamW \cite{AdamW} with learning rate $10^{-4}$, selective weight decay $10^{-4}$, cosine annealing over $200$ epochs, batch size $32$, and gradient clipping at $1.0$. The implicit layer uses $\kappa=0.8$ with spectral clipping of $W$ to enforce $\|W\|_2 \le \kappa / \|\tilde A\|_2$, and TorchDEQ \cite{TorchDEQ} with Anderson acceleration, maximum $300$ forward iterations at tolerance $10^{-6}$, and maximum $150$ backward Anderson iterations at tolerance $10^{-5}$.

All quantum modules use $n_q=4$ qubits and one Deep XYZ repetition. For SD and BD, the residual scale is $\alpha=0.1$; for ID, the quantum module is evaluated once before the DEQ solve and held fixed throughout. All quantum results use the noiseless in-code PyTorch state-vector simulator on a single NVIDIA H100 GPU. Full implementation details, hyperparameters, software versions, experiment-specific configurations, and additional diagnostics, including gradient-landscape measurements and broader train--test gap summaries, are provided in Appendix \ref{app:repro} and the released repository.

\subsection{Injection Pathway Performance Comparison}

\begin{table*}[!t]
    \captionsetup{justification=raggedright,singlelinecheck=false}
    \caption{Comparison of injection pathways across NCI1, PROTEINS, and MUTAG. Test accuracy is reported as mean $\pm$ standard deviation across the reported runs, and forward solver iterations and training time are averaged over the same runs. Boldface denotes the best value within each dataset: highest mean accuracy and lowest mean iterations/time.}
    \label{tab:pathway-main}
    \small
    \setlength{\tabcolsep}{5.5pt}
    \begin{tabular}{@{}lccccccccc@{}}
    \toprule
    \textbf{Variant} & \multicolumn{3}{c}{\textbf{NCI1}} & \multicolumn{3}{c}{\textbf{PROTEINS}} & \multicolumn{3}{c}{\textbf{MUTAG}} \\
    \cmidrule(lr){2-4}\cmidrule(lr){5-7}\cmidrule(lr){8-10}
    & \textbf{Acc.} & \textbf{Iter.} & \textbf{Time (m)} 
    & \textbf{Acc.} & \textbf{Iter.} & \textbf{Time (m)} 
    & \textbf{Acc.} & \textbf{Iter.} & \textbf{Time (m)} \\
    \midrule
    Classical equilibrium baseline 
    & $74.0 \pm 2.0$ & $81.01$ & $\mathbf{92.24}$ 
    & $75.5 \pm 4.2$ & $21.76$ & $\mathbf{22.06}$ 
    & $72.5 \pm 9.5$ & $32.99$ & $5.53$ \\
    
    State-dependent injection
    & $74.1 \pm 2.2$ & $162.01$ & $1144.82$ 
    & $75.7 \pm 5.2$ & $23.31$ & $131.03$ 
    & $73.2 \pm 8.9$ & $35.58$ & $16.28$ \\
    
    Backbone-dependent injection
    & $74.0 \pm 2.4$ & $111.94$ & $879.09$ 
    & $75.6 \pm 4.6$ & $23.22$ & $70.71$ 
    & $73.0 \pm 8.9$ & $34.97$ & $16.71$ \\
    
    Independent injection
    & $\mathbf{75.1 \pm 2.0}$ & $\mathbf{43.72}$ & $128.57 $ 
    & $\mathbf{76.9 \pm 4.7}$ & $\mathbf{21.31}$ & $24.27 $ 
    & $\mathbf{88.7 \pm 4.2}$ & $\mathbf{32.69}$ & $\mathbf{4.19}$ \\
    \bottomrule
    \end{tabular}
\end{table*}

Table \ref{tab:pathway-main} reports predictive accuracy and solver cost together across all three datasets and pathways. Independent injection is the only variant that achieves both the best test accuracy and the lowest average forward solver iterations on all three benchmarks. Relative to the classical equilibrium baseline, it reduces average forward solver iterations by $46.0\% / 2.1\% / 0.9\%$ on NCI1 / PROTEINS / MUTAG. The largest separation appears on NCI1, where iterations drop from $81.01$ to $43.72$ while test accuracy improves from $74.0\pm 2.0$ to $75.1\pm 2.0$. The same qualitative ordering persists on PROTEINS and MUTAG: both dependent pathways remain more solver-intensive than ID injection without offering a compensating accuracy gain.

\begin{figure}[t]
    \centering
    \includegraphics[width= 0.9\columnwidth]{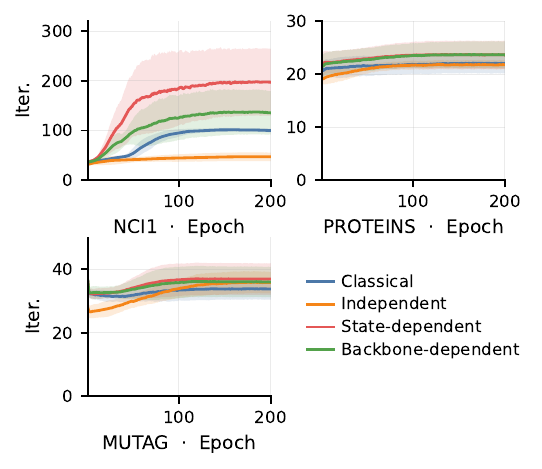}
    \caption{Training-time forward-solver effort on NCI1, PROTEINS, and MUTAG. Each panel shows mean forward DEQ solver iterations per epoch for the classical baseline and three quantum injection pathways; shaded bands denote one standard deviation across reported runs. Lower curves indicate cheaper fixed-point solves.}
    \label{fig:stability}
\end{figure}

\subsection{Solver Behavior and Theory Alignment}

Table~\ref{tab:pathway-main} and Fig.~\ref{fig:stability} show a consistent ordering in forward solver effort: ID requires the fewest iterations to achieve tolerance, SD the most, and BD generally lies between them. This ordering is most pronounced on NCI1, but the same qualitative pattern appears across all three datasets. This pattern is consistent with the contraction analysis in \S\ref{sec:contraction}. Remark~\ref{rem:ordering} gives the matched-budget Lipschitz contraction ordering:

\begin{equation}
    \kappa \le \kappa(1+\alpha L_q) \le \kappa+\alpha L_q
\end{equation}

\noindent corresponding respectively to ID, BD, and SD. While these quantities are not measured solver iterations, their ordering gives a useful qualitative explanation of solver effort: ID leaves the state-dependent part of the recurrence with the same Lipschitz budget as the classical backbone, whereas BD and SD introduce additional state-dependent quantum terms that tighten the contraction budget. Operators with more contraction headroom can reach a fixed residual tolerance in fewer solver steps, so the observed solver-cost ordering aligns with the theory without requiring the bounds to be tight.

\section{Discussion and Conclusions}\label{sec:discussion}

Three observations from \S\ref{sec:experiments} drive our discussion. First, ID injection is the only pathway that improves test accuracy over the classical equilibrium baseline on all three benchmarks. Second, the empirical ordering of forward solver effort across pathways (ID $<$ BD $<$ SD) matches the matched-budget ordering of the contraction bounds in Remark \ref{rem:ordering}. Third, both dependent pathways achieve accuracies broadly comparable to the classical baseline (albeit at the cost of higher wall-clock training time). We discuss what these observations suggest mechanistically, what they imply for hybrid quantum-classical equilibrium operator design, what the present study does and does not establish, and useful follow-ups.

\paragraph*{Interpretation of Injection Pathway Performance} ID injection couples the quantum signal to $\Phi_\Theta$ through a route the solver does not propagate state perturbations through, because the quantum term is fixed with respect to $Z$. The signal is baked in as a constant bias and therefore affects the \textit{location} of $Z^\star$ without altering the recurrence around it. SD and BD injection, by contrast, insert the quantum module into the recurrence itself, so that perturbations to $Z^{(t)}$ propagate through the quantum stage at every solver step. In SD, the quantum residual sees the raw iterate, so iterate-level noise enters $q_\eta$ directly; in BD, the classical backbone damps perturbations by a factor $\kappa$ before they reach $q_\eta$. The contraction analysis gives only sufficient conditions for a unique fixed point, which we do not verify hold for the trained models; nevertheless, the fact that the empirical test accuracy and solver-effort ordering matches the Lipschitz ordering is suggestive.

\begin{figure*}[t]
    \centering
    \includegraphics[
        width=\textwidth
    ]{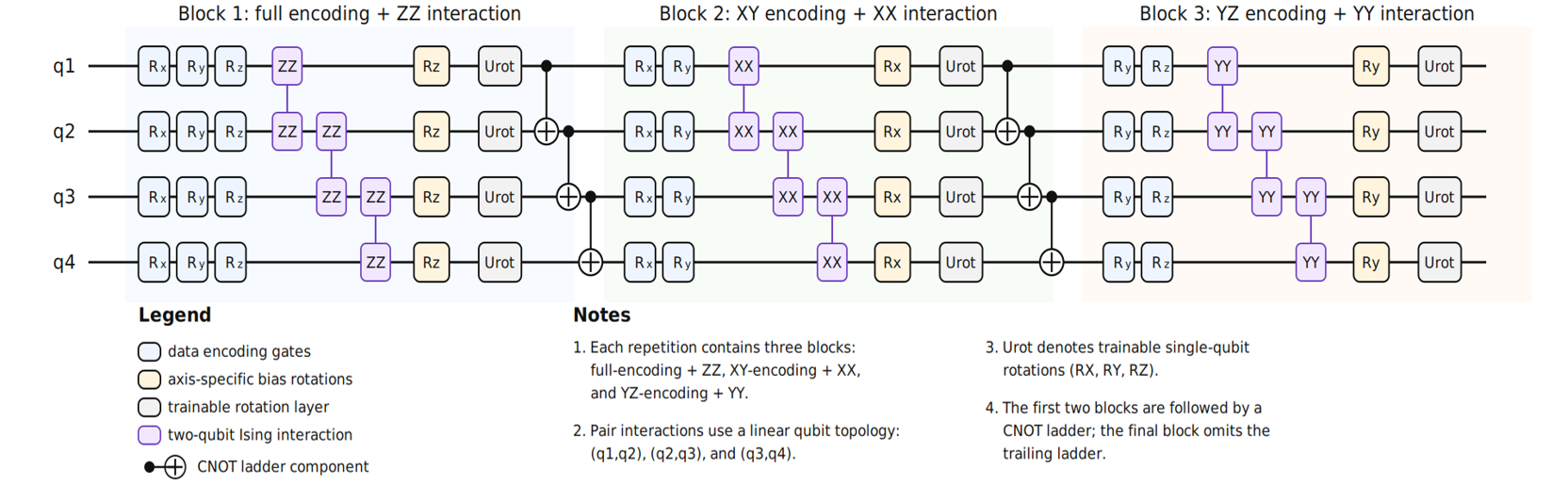}
    \caption{Deep XYZ Ansatz used in the reported experiments. The appendix schematic shows one repetition of the circuit, combining three data-dependent encoding blocks, Ising-type interaction blocks, and trainable single-qubit rotation layers.}
    \label{fig:ansatz-appendix}
\end{figure*}

\paragraph*{Implications for the Design of Injections} A conservative reading is that, when coupling an auxiliary module (quantum or otherwise) to an equilibrium operator, iterate-independent conditioning should be the default; iterate-dependent designs should be justified by a need for that behavior. The reason is twofold. First, evaluating the quantum module only once per forward solve rather than once per solver step makes ID cheaper in measurement shots and parameter-shift gradient calls than either dependent variant at the same circuit depth. Second, ID adds no $\partial Q/\partial Z$ term to the fixed-point Jacobian, so the solver does not repeatedly compose a quantum state-dependent residual during the forward solve, which may help keep optimization away from barren-plateaus \cite{BarrenPlateauOriginal,Ragone2024,Fontana2024}. Gradient diagnostics in the repository provide supporting evidence.

\paragraph*{Limitations and Future Work} This study is operates in a simulated, small qubit number regime. Three additional limitations are worth mentioning. (i) We report no non-quantum node-level static-conditioning baseline matched to the ID mechanism; without it we cannot cleanly separate the benefit of \textit{iterate-independent conditioning as a design} from the benefit of the \textit{quantum module} specifically. (ii) Our dependent injections are supported by weaker contraction conditions than ID injection; sharper contraction analyses (\textit{e.g.}, via monotone-operator splitting \cite{MonotoneIGNN,Baker2023MonotoneIGNN}) have not been attempted. (iii) Broader diagnostic runs show nontrivial train--test gaps on some datasets, indicating that regularization and model selection remain open issues orthogonal to the pathway choice; these diagnostics are provided in the repository.

\section{Conclusions} We have proposed and analyzed three operator-level pathways for coupling an encode--unitary--measure quantum module to an implicit graph equilibrium operator. We have established contraction guarantees for each, and shown empirically that the pathways separate along both accuracy and solver cost in a way consistent with our theoretical analysis. Our results indicate that \textit{where} an auxiliary signal enters a fixed-point map is a first-class design choice.

\appendices
\renewcommand{\thesubsection}{\thesection.\alph{subsection}}

\section{Implementation and Reproducibility Details}
\label{app:repro}

Code, configuration files, environment specifications, and scripts for reproducing the reported experiments are publicly available at \url{https://github.com/cxMoonGlade/QIP_IGNN}. The repository includes the training script (\texttt{train.py}), argument parsing and model construction (\texttt{model\_factory.py}), the \texttt{qignn} package implementing the three injection operators of \S\ref{sec:pathways}, a requirements file specifying the pinned dependencies (PyTorch $\ge 2.0$, PyTorch Geometric, TorchDEQ for the implicit fixed-point solver), and the command lines used to reproduce the main Table~\ref{tab:pathway-main} and Fig.~\ref{fig:stability} results. Additional figures and data, including the barren-plateau diagnostic and broader generalization diagnostics, are provided in the repository.

\section{Deep XYZ Ansatz Used in Experiments}
\label{app:ansatz}

Standard PQC templates apply the data-dependent encoding unitary $S_{u_v}$ once and process the resulting state through a single trainable block $U(\eta)$ \cite{VQA,PQCReview2019}. The \textit{data re-uploading} design of P\'erez-Salinas \textit{et al.} \cite{DRQNN} interleaves encoding and training instead, re-encoding the same $u_v$ at multiple depths with a trainable layer after each. For angle-encoded inputs, each additional re-encoding enlarges the set of Fourier frequencies the measured observables can represent as functions of $u_v$, so the function class expressible through $q_\eta$ grows with the number of re-uploadings not only with trainable depth.

The Deep XYZ ansatz of Fig.~\ref{fig:ansatz-appendix} realizes this principle through three structured blocks, broadening the generators that carry $u_v$ into the state. Block~1 applies full $(R_x,R_y,R_z)$ rotations, Block~2 only $(R_x,R_y)$, and Block~3 only $(R_y,R_z)$, and the three blocks employ distinct Ising entanglers ($ZZ$, $XX$, $YY$) on a linear qubit topology so that the trainable correlations span all three Pauli axes rather than one. In our experiments, each active quantum module uses one repetition of this template ($R=1$). The ansatz thus enlarges the expressivity of $q_\eta$ at modest qubit count and per-repetition depth, consistent with the shallow-circuit motivation of the QDEQ viewpoint adopted in \S\ref{sec:intro-qml}.

\section*{Acknowledgments}

\small{PX, TZ, LFM, and HAM were funded by a National Sciences and Engineering Research Council of Canada (NSERC) Collaborative Research and Training Experience (CREATE) grant on Quantum Computing, NSERC Alliance Consortium Grant entitled Quantum Software Consortium -- Exploring Distributed Quantum Solutions for Canada (QSC), and NSERC Alliance grant on Quantum Computing for Optimal Mobility. The authors gratefully acknowledge Yewei Wang for providing computational resources that enabled the experiments in this study.}

\newpage

\bibliographystyle{IEEEtran}
\bibliography{refs}

@article{ParameterShift,
  author  = {Schuld, Maria and Bergholm, Ville and Gogolin, Christian and Izaac, Josh and Killoran, Nathan},
  title   = {Evaluating analytic gradients on quantum hardware},
  journal = {Physical Review A},
  volume  = {99},
  number  = {3},
  pages   = {032331},
  year    = {2019},
  doi     = {10.1103/PhysRevA.99.032331},
  url     = {https://doi.org/10.1103/PhysRevA.99.032331}
}

@article{GNNreview18,
  title   = {Graph Neural Networks: A Review of Methods and Applications},
  author  = {Zhou, Jie and Cui, Ganqu and Hu, Shengding and Zhang, Zhengyan and Yang, Cheng and Liu, Zhiyuan and Wang, Lifeng and Li, Changcheng and Sun, Maosong},
  journal = {AI Open},
  volume  = {1},
  pages   = {57--81},
  year    = {2020},
  doi     = {10.1016/j.aiopen.2021.01.001},
  url     = {https://doi.org/10.1016/j.aiopen.2021.01.001}
}

@article{GNNreview24,
  author    = {Corso, Gabriele and Stark, Hannes and Jegelka, Stefanie and Jaakkola, Tommi and Barzilay, Regina},
  title     = {Graph neural networks},
  journal   = {Nature Reviews Methods Primers},
  year      = {2024},
  month     = {mar},
  volume    = {4},
  number    = {1},
  pages     = {17},
  doi       = {10.1038/s43586-024-00294-7},
  url       = {https://doi.org/10.1038/s43586-024-00294-7},
  issn      = {2662-8449}
}

@inproceedings{GCN,
  author = {Kipf, Thomas N. and Welling, Max},
  title = {Semi-Supervised Classification with Graph Convolutional Networks},
  booktitle = {International Conference on Learning Representations (ICLR)},
  year = {2017},
  url = {https://doi.org/10.48550/arXiv.1609.02907}
}

@misc{GraphSAGE,
  title={Inductive Representation Learning on Large Graphs}, 
  author={William L. Hamilton and Rex Ying and Jure Leskovec},
  year={2018},
  eprint={1706.02216},
  archivePrefix={arXiv},
  primaryClass={cs.SI},
  url={https://doi.org/10.48550/arXiv.1706.02216}
}

@misc{InductiveQGNN,
      title={Inductive Graph Representation Learning with Quantum Graph Neural Networks}, 
      author={Arthur M. Faria and Ignacio F. Graña and Savvas Varsamopoulos},
      year={2026},
      eprint={2503.24111},
      archivePrefix={arXiv},
      primaryClass={quant-ph},
      url={https://arxiv.org/abs/2503.24111}, 
}

@inproceedings{GraphQNTK,
  author = {Tang, Yehui and Yan, Junchi},
  booktitle = {Advances in Neural Information Processing Systems},
  editor = {S. Koyejo and S. Mohamed and A. Agarwal and D. Belgrave and K. Cho and A. Oh},
  pages = {6104--6118},
  publisher = {Curran Associates, Inc.},
  title = {{GraphQNTK}: Quantum Neural Tangent Kernel for Graph Data},
  doi = {10.5555/3600270.3600712},
  url = {https://doi.org/10.5555/3600270.3600712},
  volume = {35},
  year = {2022}
}

@inproceedings{GIN,
  title={How Powerful are Graph Neural Networks?},
  author={Xu, Keyulu and Hu, Weihua and Leskovec, Jure and Jegelka, Stefanie},
  booktitle={International Conference on Learning Representations (ICLR)},
  year={2019},
  url={https://doi.org/10.48550/arXiv.1810.00826}
}

@inproceedings{OonoSuzuki2020,
  title={Graph Neural Networks Exponentially Lose Expressive Power for Node Classification},
  author={Oono, Kenta and Suzuki, Taiji},
  booktitle={International Conference on Learning Representations (ICLR)},
  year={2020},
  url={https://openreview.net/forum?id=S1ldO2EFPr}
}

@inproceedings{DEQ,
  title={Deep Equilibrium Models},
  author={Bai, Shaojie and Kolter, J. Zico and Koltun, Vladlen},
  booktitle={Advances in Neural Information Processing Systems},
  year={2019},
  url={https://doi.org/10.48550/arXiv.1909.01377},
}

@inproceedings{IGNN,
  author = {Gu, Fangda and Chang, Heng and Zhu, Wenwu and Sojoudi, Somayeh and El Ghaoui, Laurent},
  booktitle = {Advances in Neural Information Processing Systems},
  editor = {H. Larochelle and M. Ranzato and R. Hadsell and M.F. Balcan and H. Lin},
  pages = {11984--11995},
  publisher = {Curran Associates, Inc.},
  title = {Implicit Graph Neural Networks},
  url = {https://proceedings.neurips.cc/paper_files/paper/2020/file/8b5c8441a8ff8e151b191c53c1842a38-Paper.pdf},
  volume = {33},
  year = {2020}
}

@inproceedings{Baker2023MonotoneIGNN,
  title={Implicit Graph Neural Networks: A Monotone Operator Viewpoint},
  author={Baker, Justin and Wang, Qingsong and Hauck, Cory D and Wang, Bao},
  booktitle={Proceedings of the 40th International Conference on Machine Learning},
  series={Proceedings of Machine Learning Research},
  volume={202},
  pages={1521--1548},
  year={2023},
  publisher={PMLR},
  url={https://proceedings.mlr.press/v202/baker23a.html}
}

@misc{IDL,
  title={Implicit Deep Learning}, 
  author={Laurent El Ghaoui and Fangda Gu and Bertrand Travacca and Armin Askari and Alicia Y. Tsai},
  year={2020},
  eprint={1908.06315},
  archivePrefix={arXiv},
  primaryClass={cs.LG},
  url={https://doi.org/10.48550/arXiv.1908.06315},
}

@inproceedings{EIGNN,
  title     = {EIGNN: Efficient Infinite-Depth Graph Neural Networks},
  author    = {Liu, Juncheng and Kawaguchi, Kenji and Hooi, Bryan and Wang, Yiwei and Xiao, Xiaokui},
  booktitle = {Advances in Neural Information Processing Systems},
  volume    = {34},
  pages     = {18762--18773},
  year      = {2021},
  url       = {https://proceedings.neurips.cc/paper/2021/hash/9bd5ee6fe55aaeb673025dbcb8f939c1-Abstract.html}
}

@inproceedings{MonotoneIGNN,
  title     = {Monotone operator equilibrium networks},
  author    = {Winston, Ezra and Kolter, J. Zico},
  booktitle = {Advances in Neural Information Processing Systems},
  volume    = {33},
  pages     = {10718--10728},
  year      = {2020},
  url       = {https://proceedings.neurips.cc/paper_files/paper/2020/hash/798d1c2813cbdf8bcdb388db0e32d496-Abstract.html}
}

@misc{MGNNI,
  title={MGNNI: Multiscale Graph Neural Networks with Implicit Layers}, 
  author={Juncheng Liu and Bryan Hooi and Kenji Kawaguchi and Xiaokui Xiao},
  year={2022},
  eprint={2210.08353},
  archivePrefix={arXiv},
  primaryClass={cs.LG},
  url={https://doi.org/10.48550/arXiv.2210.08353},
}

@misc{GIND,
  title={Optimization-Induced Graph Implicit Nonlinear Diffusion}, 
  author={Qi Chen and Yifei Wang and Yisen Wang and Jiansheng Yang and Zhouchen Lin},
  year={2022},
  eprint={2206.14418},
  archivePrefix={arXiv},
  primaryClass={cs.LG},
  url={https://doi.org/10.48550/arXiv.2206.14418},
}

@InProceedings{CerDEQ2022,
  title = 	 {{C}er{DEQ}: Certifiable Deep Equilibrium Model},
  author =       {Li, Mingjie and Wang, Yisen and Lin, Zhouchen},
  booktitle = 	 {Proceedings of the 39th International Conference on Machine Learning},
  pages = 	 {12998--13013},
  year = 	 {2022},
  editor = 	 {Chaudhuri, Kamalika and Jegelka, Stefanie and Song, Le and Szepesvari, Csaba and Niu, Gang and Sabato, Sivan},
  volume = 	 {162},
  series = 	 {Proceedings of Machine Learning Research},
  month = 	 {17--23 Jul},
  publisher =    {PMLR},
  url = 	 {https://proceedings.mlr.press/v162/li22t.html}
}

@inproceedings{LyaDEQ2023,
  title     = {Lyapunov-Stable Deep Equilibrium Models},
  author    = {Chu, Haoyu and Wei, Shikui and Liu, Ting and Zhao, Yao and Miyatake, Yuto},
  booktitle = {Proceedings of the AAAI Conference on Artificial Intelligence},
  volume    = {38},
  number    = {10},
  pages     = {11615--11623},
  year      = {2024},
  doi       = {10.1609/aaai.v38i10.29044},
  url       = {https://doi.org/10.1609/aaai.v38i10.29044},
}

@misc{TorchDEQ,
  title={TorchDEQ: A Library for Deep Equilibrium Models}, 
  author={Zhengyang Geng and J. Zico Kolter},
  year={2023},
  eprint={2310.18605},
  archivePrefix={arXiv},
  primaryClass={cs.LG},
  url={https://doi.org/10.48550/arXiv.2310.18605},
}

@article{PQCReview2019,
  title={Parameterized quantum circuits as machine learning models},
  volume={4},
  ISSN={2058-9565},
  doi={10.1088/2058-9565/ab4eb5},
  url={https://doi.org/10.1088/2058-9565/ab4eb5},
  number={4},
  journal={Quantum Science and Technology},
  publisher={IOP Publishing},
  author={Benedetti, Marcello and Lloyd, Erika and Sack, Stefan and Fiorentini, Mattia},
  year={2019},
  month=nov,
  pages={043001},
}

@article{DRQNN,
  title={Data re-uploading for a universal quantum classifier},
  volume={4},
  ISSN={2521-327X},
  doi={10.22331/q-2020-02-06-226},
  url={https://doi.org/10.22331/q-2020-02-06-226},
  journal={Quantum},
  publisher={Verein zur Forderung des Open Access Publizierens in den Quantenwissenschaften},
  author={Pérez-Salinas, Adrián and Cervera-Lierta, Alba and Gil-Fuster, Elies and Latorre, José I.},
  year={2020},
  month=feb,
  pages={226},
}

@article{QuantumNeuralNetworks,
  title={The power of quantum neural networks},
  volume={1},
  ISSN={2662-8457},
  doi={10.1038/s43588-021-00084-1},
  url={https://doi.org/10.1038/s43588-021-00084-1},
  number={6},
  journal={Nature Computational Science},
  publisher={Springer Science and Business Media LLC},
  author={Abbas, Amira and Sutter, David and Zoufal, Christa and Lucchi, Aurelien and Figalli, Alessio and Woerner, Stefan},
  year={2021},
  month=jun,
  pages={403--409},
}

@article{VQA,
  title={Variational quantum algorithms},
  volume={3},
  ISSN={2522-5820},
  doi={10.1038/s42254-021-00348-9},
  url={https://doi.org/10.1038/s42254-021-00348-9},
  number={9},
  journal={Nature Reviews Physics},
  publisher={Springer Science and Business Media LLC},
  author={Cerezo, M. and Arrasmith, Andrew and Babbush, Ryan and Benjamin, Simon C. and Endo, Suguru and Fujii, Keisuke and McClean, Jarrod R. and Mitarai, Kosuke and Yuan, Xiao and Cincio, Lukasz and Coles, Patrick J.},
  year={2021},
  month=aug,
  pages={625--644},
}

@article{BarrenPlateauOriginal,
  title={Barren plateaus in quantum neural network training landscapes},
  volume={9},
  ISSN={2041-1723},
  doi={10.1038/s41467-018-07090-4},
  url={https://doi.org/10.1038/s41467-018-07090-4},
  number={1},
  journal={Nature Communications},
  publisher={Springer Science and Business Media LLC},
  author={McClean, Jarrod R. and Boixo, Sergio and Smelyanskiy, Vadim N. and Babbush, Ryan and Neven, Hartmut},
  year={2018},
  month=nov,
}

@article{Ragone2024,
  title={A Lie algebraic theory of barren plateaus for deep parameterized quantum circuits},
  volume={15},
  ISSN={2041-1723},
  doi={10.1038/s41467-024-49909-3},
  url={https://doi.org/10.1038/s41467-024-49909-3},
  number={1},
  journal={Nature Communications},
  publisher={Springer Science and Business Media LLC},
  author={Ragone, Michael and Bakalov, Bojko N. and Sauvage, Frédéric and Kemper, Alexander F. and Ortiz Marrero, Carlos and Larocca, Martín and Cerezo, M.},
  year={2024},
  month=aug,
}

@article{Fontana2024,
  title={Characterizing barren plateaus in quantum ans{\"a}tze with the adjoint representation},
  volume={15},
  ISSN={2041-1723},
  doi={10.1038/s41467-024-49910-w},
  url={https://doi.org/10.1038/s41467-024-49910-w},
  number={1},
  journal={Nature Communications},
  publisher={Springer Science and Business Media LLC},
  author={Fontana, Enrico and Herman, Dylan and Chakrabarti, Shouvanik and Kumar, Niraj and Yalovetzky, Romina and Heredge, Jamie and Sureshbabu, Shree Hari and Pistoia, Marco},
  year={2024},
  month=aug,
}

@misc{QGNN,
  title        = {Quantum Graph Neural Networks},
  author       = {Verdon, Guillaume and McCourt, Trevor and Luzhnica, Enxhell and Singh, Vikash and Leichenauer, Stefan and Hidary, Jack},
  year         = {2019},
  eprint       = {1909.12264},
  archivePrefix= {arXiv},
  primaryClass = {quant-ph},
  note         = {arXiv preprint},
  url          = {https://doi.org/10.48550/arXiv.1909.12264},
}

@misc{QGNNreview24,
  title={From Graphs to Qubits: A Critical Review of Quantum Graph Neural Networks}, 
  author={Andrea Ceschini and Francesco Mauro and Francesca De Falco and Alessandro Sebastianelli and Alessio Verdone and Antonello Rosato and Bertrand Le Saux and Massimo Panella and Paolo Gamba and Silvia L. Ullo},
  year={2024},
  eprint={2408.06524},
  archivePrefix={arXiv},
  primaryClass={quant-ph},
  url={https://doi.org/10.48550/arXiv.2408.06524},
}

@misc{QGCN,
  title={Quantum Graph Convolutional Neural Networks}, 
  author={Jin Zheng and Qing Gao and Yanxuan Lv},
  year={2021},
  eprint={2107.03257},
  archivePrefix={arXiv},
  primaryClass={eess.SP},
  url={https://doi.org/10.48550/arXiv.2107.03257},
}

@inproceedings{Hu2022QuGCNDesign,
  title={On the Design of Quantum Graph Convolutional Neural Network in the NISQ-Era and Beyond},
  author={Hu, Zhirui and Li, Jinyang and Pan, Zhenyu and Zhou, Shanglin and Yang, Lei and Ding, Caiwen and Khan, Omer and Geng, Tong and Jiang, Weiwen},
  booktitle={Proceedings - 2022 IEEE 40th International Conference on Computer Design, ICCD 2022},
  pages={290--297},
  year={2022},
  doi={10.1109/ICCD56317.2022.00050},
  url={https://doi.org/10.1109/ICCD56317.2022.00050},
}

@inproceedings{GQGLA,
  title={On Designing General and Expressive Quantum Graph Neural Networks with Applications to {MILP} Instance Representation},
  author={Xinyu Ye and Hao Xiong and Jianhao Huang and Ziang Chen and Jia Wang and Junchi Yan},
  booktitle={The Thirteenth International Conference on Learning Representations},
  year={2025},
  url={https://openreview.net/forum?id=IQi8JOqLuv},
}

@misc{HaQGNN,
  title={HaQGNN: Hardware-Aware Quantum Kernel Design Based on Graph Neural Networks}, 
  author={Yuxiang Liu and Fanxu Meng and Lu Wang and Yi Hu and Sixuan Li and Xutao Yu and Zaichen Zhang},
  year={2025},
  eprint={2506.21161},
  archivePrefix={arXiv},
  primaryClass={quant-ph},
  url={https://doi.org/10.48550/arXiv.2506.21161},
}

@inproceedings{QDEQ,
  title={Quantum Deep Equilibrium Models},
  author={Philipp Schleich and Marta Skreta and Lasse Bj{\o}rn Kristensen and Rodrigo Vargas-Hernandez and Alan Aspuru-Guzik},
  booktitle={The Thirty-eighth Annual Conference on Neural Information Processing Systems},
  year={2024},
  url={https://openreview.net/forum?id=CWhwKb0Q4k},
}

@misc{ResHQCNN,
  title={A Hybrid Quantum-Classical Neural Network with Deep Residual Learning},
  author={Yanying Liang and Wei Peng and Zhu-Jun Zheng and Olli Silv{\'e}n and Guoying Zhao},
  year={2021},
  eprint={2012.07772},
  archivePrefix={arXiv},
  primaryClass={cs.LG},
  url={https://doi.org/10.48550/arXiv.2012.07772},
}

@article{ResQNets,
  title={ResQNets: A Residual Approach for Mitigating Barren Plateaus in Quantum Neural Networks},
  author={Muhammad Kashif and Saif Al-kuwari},
  journal={EPJ Quantum Technology},
  volume={11},
  pages={4},
  year={2024},
  doi={10.1140/epjqt/s40507-023-00216-8},
  url={https://doi.org/10.1140/epjqt/s40507-023-00216-8},
}

@inproceedings{QCE22Embedding,
  title={Embedding Learning in Hybrid Quantum-Classical Neural Networks},
  doi={10.1109/QCE53715.2022.00026},
  url={https://doi.org/10.1109/QCE53715.2022.00026},
  booktitle={2022 IEEE International Conference on Quantum Computing and Engineering (QCE)},
  publisher={IEEE},
  author={Liu, Minzhao and Liu, Junyu and Liu, Rui and Makhanov, Henry and Lykov, Danylo and Apte, Anuj and Alexeev, Yuri},
  year={2022},
  month=sep,
  pages={79--86},
}

@INPROCEEDINGS{QNetGAN,
  author={Cui, Max and Chang, Linda and Chau, Adelina and Mekuria, Hasset and Adwankar, Leena and Pendyala, Sriaditya and McMahan, Larry},
  booktitle={2024 IEEE International Conference on Quantum Computing and Engineering (QCE)}, 
  title={Efficient and Optimized Small Organic Molecular Graph Generation Pathway Using a Quantum Generative Adversarial Network}, 
  year={2024},
  volume={01},
  pages={1565--1570},
  doi={10.1109/QCE60285.2024.00183},
  url={https://doi.org/10.1109/QCE60285.2024.00183},
}

@INPROCEEDINGS{QCE24HQCGNN,
  author={Ray, Anupama and Madan, Dhiraj and Patil, Srushti and Pati, Pushpak and Rapsomaniki, Marianna and Kohlakala, Aviwe and Dlamini, Thembelihle Rose and Muller, Stephanie Julia and Rhrissorrakrai, Kahn and Utro, Filippo and Parida, Laxmi},
  booktitle={2024 IEEE International Conference on Quantum Computing and Engineering (QCE)}, 
  title={Hybrid Quantum-Classical Graph Neural Networks for Tumor Classification in Digital Pathology}, 
  year={2024},
  volume={01},
  pages={1611--1616},
  doi={10.1109/QCE60285.2024.00188},
  url={https://doi.org/10.1109/QCE60285.2024.00188},
}

@inproceedings{QCE24GNN_EE,
  title={Graph Neural Networks for Parameterized Quantum Circuits Expressibility Estimation},
  doi={10.1109/QCE60285.2024.00181},
  url={https://doi.org/10.1109/QCE60285.2024.00181},
  booktitle={2024 IEEE International Conference on Quantum Computing and Engineering (QCE)},
  publisher={IEEE},
  author={Aktar, Shamminuj and Bärtschi, Andreas and Oyen, Diane and Eidenbenz, Stephan and Badawy, Abdel-Hameed A.},
  year={2024},
  month=sep,
  pages={1547--1557},
}

@inproceedings{QCE24OH_GNN,
  title={On Optimizing Hyperparameters for Quantum Neural Networks},
  doi={10.1109/QCE60285.2024.00174},
  url={https://doi.org/10.1109/QCE60285.2024.00174},
  booktitle={2024 IEEE International Conference on Quantum Computing and Engineering (QCE)},
  publisher={IEEE},
  author={Herbst, Sabrina and De Maio, Vincenzo and Brandic, Ivona},
  year={2024},
  month=sep,
  pages={1478--1489},
}

@inproceedings{Transformer,
  title = {Attention Is All You Need},
  author = {Vaswani, Ashish and Shazeer, Noam and Parmar, Niki and Uszkoreit, Jakob and Jones, Llion and Gomez, Aidan N. and Kaiser, Lukasz and Polosukhin, Illia},
  booktitle = {Advances in Neural Information Processing Systems (NeurIPS)},
  volume = {30},
  pages = {5998--6008},
  year = {2017},
  url={https://doi.org/10.48550/arXiv.1706.03762},
}

@article{FiLM, 
  title={FiLM: Visual Reasoning with a General Conditioning Layer}, 
  volume={32}, 
  url={https://doi.org/10.1609/aaai.v32i1.11671}, 
  DOI={10.1609/aaai.v32i1.11671}, 
  number={1}, 
  journal={Proceedings of the AAAI Conference on Artificial Intelligence}, 
  author={Perez, Ethan and Strub, Florian and de Vries, Harm and Dumoulin, Vincent and Courville, Aaron}, 
  year={2018}, 
  month={Apr.} ,
}

@misc{AdamW,
  title={Decoupled Weight Decay Regularization}, 
  author={Ilya Loshchilov and Frank Hutter},
  year={2019},
  eprint={1711.05101},
  archivePrefix={arXiv},
  primaryClass={cs.LG},
  url={https://doi.org/10.48550/arXiv.1711.05101},
}

@inproceedings{TUDatasets,
  title={TUDataset: A collection of benchmark datasets for learning with graphs},
  author={Christopher Morris and Nils M. Kriege and Franka Bause and Kristian Kersting and Petra Mutzel and Marion Neumann},
  booktitle={ICML 2020 Workshop on Graph Representation Learning and Beyond (GRL+ 2020)},
  archivePrefix={arXiv},
  eprint={2007.08663},
  url={https://doi.org/10.48550/arXiv.2007.08663},
  year={2020},
}

@inproceedings{FairComparisonGNNGraphClass,
  title={A Fair Comparison of Graph Neural Networks for Graph Classification},
  author={Errica, Federico and Podda, Marco and Bacciu, Davide and Micheli, Alessio},
  booktitle={International Conference on Learning Representations (ICLR)},
  year={2020},
  url={https://openreview.net/forum?id=HygDF6NFPB},
}

@misc{BenchmarkingGNNs,
  title={Benchmarking Graph Neural Networks},
  author={Dwivedi, Vijay Prakash and Joshi, Chaitanya K. and Luu, Anh Tuan and Laurent, Thomas and Bengio, Yoshua and Bresson, Xavier},
  year={2020},
  eprint={2003.00982},
  archivePrefix={arXiv},
  primaryClass={cs.LG},
  url={https://doi.org/10.48550/arXiv.2003.00982},
}

@book{Rudin1976,
  title={Principles of mathematical analysis},
  author={Rudin, Walter},
  volume={3},
  year={1976},
  publisher={McGraw-Hill}
}

@book{HornJohnson2012,
  title={Matrix Analysis},
  author={Horn, Roger A. and Johnson, Charles R.},
  edition={2nd},
  year={2012},
  publisher={Cambridge University Press},
  doi={10.1017/9781139020411},
  url={https://doi.org/10.1017/9781139020411},
}

\end{document}